\documentclass{amsart}
\usepackage{amsmath,amssymb}
\usepackage[english]{babel}

\usepackage{setspace}

\usepackage[usenames,dvipsnames]{pstricks}
\usepackage{epsfig}
\usepackage{pst-grad}
\usepackage{pst-plot}
\usepackage{multirow}
\usepackage{rotating}
\usepackage{times}

\newcommand\aes{a_{\emptyset}}

\newcommand{\beq}{\begin{equation}}
\newcommand{\eeq}{\end{equation}}
\newcommand{\bdm}{\begin{displaymath}}
\newcommand{\edm}{\end{displaymath}}
\newcommand{\bea}{\begin{eqnarray}}
\newcommand{\eea}{\end{eqnarray}}

\newcommand{\poly}{{\mbox{\rm poly}}}
\newcommand{\benum}{\begin{enumerate}}
\newcommand{\eenum}{\end{enumerate}}
\newcommand{\bit}{\begin{itemize}}
\newcommand{\eit}{\end{itemize}}
\newcommand{\bdes}{\begin{description}}
\newcommand{\edes}{\end{description}}
\newcommand{\bpic}{\begin{picture}}
\newcommand{\epic}{\end{picture}}
\newcommand{\bc}{\begin{center}}
\newcommand{\ec}{\end{center}}
\newcommand{\oc}{{\mathrm{oc}}}

\newtheorem{theorem}{Theorem}[section]
\newtheorem{proposition}[theorem]{Proposition}
\newtheorem{corollary}[theorem]{Corollary}

\newtheorem{definition}[theorem]{Definition}
\newtheorem{example}[theorem]{Example}

\def\np{{\sf NP}}

\def\GASP{{\sf GASP}}
\def\aGASP{{\sf a\text{-}GASP}}
\def\calY{{\mathcal Y}}
\def\calC{{\mathcal C}}
\def\calR{{\mathcal R}}

\begin{document}

\title{Group Activity Selection Problem}
\author{
Andreas Darmann$^1$ \and  
Edith Elkind$^2$ \and 
Sascha Kurz$^3$ \and
J\'er\^ome Lang$^4$ \and 
Joachim Schauer$^1$ \and 
Gerhard Woeginger$^5$
}

\address{
Universit\"at Graz, Austria$^1$
\and
Nanyang Technological University, Singapore$^2$
\and
Universit\"at Bayreuth, Germany$^3$
\and
LAMSADE -- Universit\'e Paris-Dauphine, France$^4$
\and
TU Eindhoven, The Netherlands$^5$
}

\maketitle

\begin{abstract}
We consider a setting where one has to organize one or several group activities 
for a set of agents. 
Each agent will participate in at most one activity, and her preferences over activities 
depend on the number of participants in the activity. 
The goal is to assign agents to activities based on their preferences.
We put forward a general model for this setting, which is a natural generalization 
of anonymous hedonic games. We then focus on a special case of our model, 
where agents' preferences are binary, i.e., each agent classifies all pairs 
of the form "(activity, group size)" into ones that are acceptable and ones that are not.  
We formulate several solution concepts for this scenario, 
and study them from the computational point of view, providing
hardness results for the general case as well as efficient algorithms
for settings where agents' preferences satisfy certain natural constraints.
\end{abstract}

\section{Introduction}
There are many real-life situations where a group of agents is faced with a choice 
of multiple activities, and the members of the group have differing preferences over these
activities. Sometimes it is feasible for the group to split into smaller subgroups, so that each subgroup
can pursue its own activity. Consider, for instance, a workshop whose organizers would like 
to arrange one or more social activities for the free afternoon.\footnote{Some of the co-authors
of this paper had to deal with this problem when co-organizing a Dagstuhl seminar.} 
The available activities include a hike, a bus trip, and a table tennis competition.
As they will take place simultaneously, 
each attendee can select at most one activity (or choose not to participate).
It is easy enough to elicit the attendees' preferences over the activities, and divide
the attendees into groups based on their choices. However, the situation
becomes more complicated if one's preferences may depend on the number 
of other attendees who choose the same activity.  
For instance, the bus trip has a fixed transportation cost that has to be shared among its participants, 
which implies that, typically, an attendee $i$ is only willing to go on the bus trip 
if the number of other participants of the bus trip exceeds a threshold $\ell_i$. 
Similarly, $i$ may only be willing to play table tennis
if the number of attendees who signed up for the tournament
does {\em not} exceed a threshold $u_i$: as there is only one table, the more 
participants, the less time each individual spends playing. 

Neglecting to take the number of participants of each activity into account may lead to highly undesirable 
outcomes, such as a bus that is shared by two persons, each of them paying a high cost, 
and a 48-participant table tennis tournament with one table. Adding constraints on the number
of participants for each activity is a practical, but imperfect
solution, as the agents' preferences over group sizes may differ: while some attendees
(say, senior faculty) may be willing to go on the bus trip with just 4--5 other participants, others
(say, graduate students) cannot afford it unless the number of participants exceeds 10.
A more fine-grained approach 
is to elicit the agents' preferences over pairs of the form ``(activity, group size)'',
rather than over activities themselves, and allocate agents to activities based on this information. 
In general, agents' preferences can be thought of as weak orders over all such pairs, including
the pair ``(do nothing, 1)'', which we will refer to as the {\em void activity}. A simpler
model, which will be the main focus of this paper,  
assumes that each agents classifies all pairs into ones that are acceptable to him and ones
that are not, and if an agent views his current assignment as unacceptable, he prefers (and is allowed)
to switch to the void activity (so the assignment is unstable unless
it is acceptable to all agents).

The problem of finding a good assignment of agents to activities, 
which we will refer to as the {\em Group Activity Selection Problem} ($\GASP$), 
may be viewed as a mechanism design problem (or, more narrowly, a voting problem) 
or as a coalition formation problem, depending on whether
we expect the agents to act strategically when reporting their preferences. Arguably, in our
motivating example the agents are likely to be honest, so throughout the paper we assume that the 
central authority knows (or, rather, can reliably elicit)
the agents' true preferences, and its goal is to find an assignment of players
to activities that, informally speaking, is stable and/or maximizes the overall satisfaction. 
This model is closely related to that of {\em anonymous hedonic games}~\cite{banerjee:1998a}, where, 
just as in our setting, players have to split into groups and each player has preferences
over possible group sizes. The main difference between anonymous hedonic games and our problem
is that, in our setting, the agents' preferences depend not only on the group size, 
but also on the activity that has been allocated to their group; thus, our model can be seen
as a generalization of anonymous hedonic games. On the other hand, we can represent our problem
as a general (i.e., non-anonymous) hedonic game~\cite{bogomolnaia:2002,banerjee:1998a}, by creating
a dummy agent for each activity and endowing it with suitable preferences 
(see Section~\ref{sec:hedonic} for details). However, 
our setting has useful structural properties that distinguish it from a generic hedonic game:
for instance, it allows for succinct representation of players' preferences, and, as we will see, 
has several natural special cases that admit efficient algorithms for finding good outcomes. 

In this paper, we initiate the formal study of $\GASP$. 
Our goal is to put forward a model for this problem that is expressive 
enough to capture many real-life activity selection scenarios, 
yet simple enough to admit efficient procedures for finding good assignments of agents
to activities. We describe the basic structure of the problem, 
and discuss plausible constraints of the number and type of available activities 
and the structure of agents' preferences.
We show that even under a fairly simple preference model (where agents are assumed
to approve or disapprove each available alternative)
finding an assignment that maximizes the number of satisfied agents is computationally
hard; however, we identify several natural special cases of the problem that admit 
efficient algorithms. We also briefly discuss the issue of stability
in our setting.

We do not aim to provide a complete analysis of the group activity selection problem; 
rather, we view our work as a first step towards understanding the algorithmic 
and incentive issues that arise in this setting.
We hope that our paper will lead
to future research on this topic; to facilitate this, throughout the paper
we highlight several possible extensions of our model as well as list some
problems left open by our work.  


\section{Formal Model}\label{model}


\begin{definition}
An instance of the {\em Group Activity Selection Problem ({\GASP})} is given 
by a set of  {\em agents}  $N = \{1,\ldots, n\}$, 
a set of {\em activities} $A=A^*\cup\{\aes\}$, where $A^* = \{a_1, \ldots, a_p\}$, 
and a {\em profile} $P$, which consists of $n$ {\em votes} (one for each agent): 
$P = (V_1,\ldots, V_n)$. The vote of agent $i$ describes his preferences
over the set of {\em alternatives} $X = X^*\cup\{\aes\}$, 
where $X^* =A^* \times \{1,\ldots,n\}$;
alternative $(a,k)$, $a\in A^*$, is interpreted as ``activity $a$ with $k$ participants'', 
and $\aes$ is the {\em void activity}. 

The vote $V_i$ of an agent $i\in N$ is (also denoted by $\succeq_i$)
is a weak order over $X^*$;
its induced strict preference and indifference relations 
are denoted by  $\succ_i$ and $\sim_i$, respectively.
We set $S_i=\{(a, k)\in X^*\mid (a, k)\succ_i \aes\}$;
we say that voter $i$ {\em approves} of all alternatives in $S_i$, 
and refer to the set $S_i$ as the {\em induced approval vote} of voter $i$.

Throughout the paper we will mostly focus on a special case of our problem
where no agent is indifferent between the void activity and any other alternative
(i.e., for any $i\in N$ we have $\{x\in X^*\mid x\sim_i\aes\}=\emptyset$), 
and each agent is indifferent between all the alternatives in $S_i$. In other words,
preferences are {\em trichotomous}: the agent partitions $X$ into three clusters
$S_i$, $\{\aes\}$ and $X \setminus (S_i \cup \{\aes\})$, is indifferent between
two alternatives of the same cluster, prefers any $(a,k)$ in $S_i$ to
$\aes$, and $\aes$ to any $(a,k)$ in $X \setminus (S_i \cup \{\aes\})$;
we denote this special case of our problem by $\aGASP$. 
\end{definition}

It will be convenient to
distinguish between activities that are unique and ones that exist in multiple copies. For instance,
if there is a single tennis table and two buses, then we can organize one table tennis tournament, 
two bus trips (we assume that there is only one potential destination for the bus trip, so these
trips are identical), and an unlimited number of hikes (again, we assume that there is only one
hiking route). This distinction will be useful for the purposes
of complexity analysis: for instance, some of the problems we consider
are easy when we have $k$ copies of one activity, but hard when we have $k$ distinct activities. 
Formally, we say that two activities $a$ and $b$ are {\em equivalent}
if for every agent $i$ and every $j\in\{1, \dots, n\}$ it holds that $(a, j)\sim_i(b, j)$.
We say that an activity $a\in A^*$ is {\em $k$-copyable} if $A^*$
contains exactly $k$ activities that are equivalent to $a$ (including $a$ itself).
We say that $a$ is {\em simple} if it is $1$-copyable; if $a$ is $k$-copyable for $k\ge n$, 
we will say that it is {\em $\infty$-copyable} (note that we would never need 
to organize more than $n$ copies of any activity). If some activities in $A^*$ are equivalent, 
$A^*$ can be represented more succinctly by listing one representative of each
equivalence class, together with the number of available copies. However, 
as long as we make the reasonable assumption that each activity exists in at most $n$
copies, this representation is at most polynomially more succinct. 

Our model can be enriched by specifying a set of {\em constraints} $\Gamma$. 
One constraint that arises frequently in practice 
is a {\em global cardinality} constraint, which specifies a bound $K$
on the number of activities to be organized. More generally, we could also consider 
more complex constraints on the set of activities that can be organized simultaneously, 
which can be encoded, e.g., by a propositional formula or a set of linear inequalities.
We remark that there can also be external constraints on the number of participants for each activity:
for instance, a bus can fit at most $40$ people. However, these constraints can be incorporated
into agents' preferences, by assuming that all agents view the alternatives
that do not satisfy these constraints as unacceptable. 


\subsection{Special Cases}\label{sec:special}

We now consider some natural restrictions on agents' preferences that may
simplify the problem of finding a good assignment. 
We first need to introduce some additional notation.
Given a vote $V_i$ and an activity $a \in A^*$, let $S_i^{\downarrow a}$
denote the projection of $S_i$ onto $\{a\} \times \{1,\ldots,n\}$. That is, we set 
$S_i^{\downarrow a} = \{k \mid (a, k) \in S_i\}$.

\begin{example}
Let $A^* = \{a,b\}$ and consider an agent $i$ whose vote $V_i$ is given by\\
\centerline{$
(a,8) \succ_i (a,7) \sim_i (b,4) \succ_i (a,9) \succ_i (b,3) \succ_i (b,5) \succ_i 
(a,6) \succ_i (b,6) \succ_i \aes  \succ_i \ldots
$}\\
Then 
$S_i=\{a\}\times[6, 9]\cup\{b\}\times[3, 6]$ and $S_i^{\downarrow a}=\{6, 7, 8, 9\}$. 
\end{example}

We are now ready to define two types of restricted preferences 
for $\aGASP$ that are directly
motivated by our running example, namely, {\em increasing} and {\em decreasing}
preferences. Informally, under increasing preferences an agent prefers to share
each activity with as many other participants as possible (e.g., because each activity
has an associated cost, which has to be split among the participants), and under
decreasing preferences an agent prefers to share each activity with as few other participants
as possible (e.g., because each activity involves sharing a limited resource).
Of course, an agent's preferences may also be increasing with respect to some activities
and decreasing with respect to others, depending on the nature of each activity.
We provide a formal definition for $\aGASP$ only; however, it can be extended
to $\GASP$ in a straightforward way.

\begin{definition}\label{def:inc-dec-mix}
Consider an instance $(N, A, P)$ of $\aGASP$.
We say that the preferences of agent $i$ are {\em increasing (INC)}
with respect to an activity $a\in A^*$ 
if there exists a threshold $\ell_i^a\in\{1, \dots, n+1\}$
such that $S_i^{\downarrow a} = [\ell_i^a, n]$ (where we assume that $[n+1, n]=\emptyset$).
Similarly, 
we say that the preferences of agent $i$ are {\em decreasing (DEC)}
with respect to an activity $a\in A^*$
if there exists a threshold $u_i^a\in\{0, \dots, n\}$
such that $S_i^{\downarrow a} = [1, u_i^a]$ (where we assume that $[1, 0]=\emptyset$).

We say that an instance $(N, A, P)$ of $\aGASP$
is {\em increasing} (respectively, {\em decreasing}) 
if the preferences of each agent $i\in N$ are increasing 
(respectively, decreasing) with respect to each activity $a\in A^*$.
We say that an instance $(N, A, P)$ of $\aGASP$
is {\em mixed increasing-decreasing (MIX)}
if there exists a set $A^+\subseteq A^*$ such that
for each agent $i\in N$ his preferences are increasing
with respect to each $a\in A^+$ and decreasing with respect 
to each $a\in A^- = A^*\setminus A^+$.
\end{definition}

A recently proposed model which can be embedded into $\GASP$ with decreasing preferences is the ordinal version of {\em cooperative group buying} (\cite{LuBoutilier12}, Section 6): the model has a set of buyers and a set of items with volume discounts; buyers rank all pairs $(j,p_j)$ for any item $j$ and any of its possible discounted prices, where the discounted price is a function of the number of buyers who are matched to the item. 

For some activities, an agent may have both a lower and an upper bound on the acceptable
group size: e.g., one may prefer to go on a hike with at least $3$ other people, but does
not want the group to be too large (so that it can maintain a good pace). 
In this case, we say that an agent has {\em interval} (INV) preferences; note that 
INC/DEC/MIX
preferences are a special case
of interval preferences. 

\begin{definition}\label{def:int}
Consider an instance $(N, A, P)$ of $\aGASP$.
We say that the preferences of agent $i$ are {\em interval (INV)}
if for each $a\in A^*$ there exists a pair of thresholds 
$\ell_i^a, u_i^a\in\{1, \dots, n\}$   
such that $S_i^{\downarrow a} = [\ell_i^a, u_i^a]$ 
(where we assume that $[i, j]=\emptyset$ for $i>j$).
\end{definition}

Other natural constraints on preferences include restricting the size of $S_i$
(or, more liberally, that of $S_i^{\downarrow a}$ for each $a\in A^*$), 
or requiring agents to have similar preferences: for instance, 
one could limit the number of agent {\em types}, i.e., require that the set
of agents can be split into a small number of groups so that the agents
in each group have identical preferences. We will not define such constraints
formally, but we will indicate if they are satisfied by the instances
constructed in the hardness proofs in Section~\ref{sec:hard}.

\subsection{$\GASP$ and Hedonic Games}\label{sec:hedonic}
Recall that a {\em hedonic game}~\cite{banerjee:1998a,bogomolnaia:2002} 
is given by a set of agents $N$, and, for each agent
$i\in N$, a weak order $\ge_i$ over all coalitions (i.e., subsets of $N$) that include $i$.
That is, in a hedonic game each agent has preferences over coalitions 
that he can be a part of. A coalition $S$, $i\in S$, is said to be {\em unacceptable}
for player $i$ if $\{i\}>_i S$. 
A hedonic game is said to be {\em anonymous} if each agent
is indifferent among all coalitions of the same size that include him, i.e., 
for every $i\in N$ and every $S, T\subseteq N\setminus\{i\}$ 
such that $|S|=|T|$ it holds that
$S\cup\{i\}\ge_i T\cup\{i\}$ and $T\cup\{i\}\ge_i S\cup\{i\}$.

At a first glance, it may seem that the $\GASP$ formalism
is more general than that of hedonic games, since in $\GASP$
the agents care not only about their coalition, 
but also about the activity they have been assigned to.
However, we will now argue that $\GASP$ can be embedded 
into the hedonic games framework.

Given an instance of the $\GASP$ problem $(N, A, P)$ with $|N|=n$, where the $i$-th
agent's preferences are given by a weak order $\succeq_i$, 
we construct a hedonic game $H(N, A, P)$ as follows.
We create $n+p$ players; the first $n$ players correspond to agents in $N$, and the last $p$
players correspond to activities in $A^*$. The last $p$ players are indifferent among all coalitions.
For each $i=1, \dots, n$, player $i$
ranks every non-singleton coalition with no activity players as unacceptable;
similarly, all coalitions with two or more activity players are ranked as unacceptable.
The preferences over coalitions with exactly one activity player are derived
naturally from the votes: if $S, T$ are two coalitions involving player $i$, $x$ is the unique activity
player in $S$, and $y$ is the unique activity player in $T$, then $i$ 
weakly prefers $S$ to $T$ in $H(N, A, P)$
if and only if  $(x, |S|-1)\succeq_i (y, |T|-1)$,
and $i$ weakly prefers $S$ to $\{i\}$ in $H(N, A, P)$ if and only if $(x, |S|-1)\succeq_i \aes$.
We emphasize that the resulting hedonic games are not anonymous.
Further, while this embedding allows us to apply the standard solution concepts
for hedonic games without redefining them, the intuition behind these solution
concepts is not always preserved (e.g., because activity players never want to deviate). 
Therefore, in Section~\ref{concepts}, we will provide formal definitions of the relevant
hedonic games solution concepts adapted to the setting of $\aGASP$.

We remark that when $A^*$ consists of a single $\infty$-copyable activity 
(i.e., there are $n$ activities in $A^*$, all of them equivalent to each other), 
$\GASP$ become equivalent to anonymous hedonic games.
Such games have been studied in detail by Ballester~\cite{ballester}, 
who provides a number of complexity results for them. In particular, 
he shows that finding an outcome that is core stable, Nash stable or individually
stable 
(see Section~\ref{concepts} for the definitions of some of these concepts in the context of $\aGASP$) 
is $\np$-hard. 
Clearly, all these complexity results also hold for $\GASP$. 
However, they do not directly imply similar hardness results for $\aGASP$.


\section{Solution Concepts}\label{concepts}
Having discussed the basic model of $\GASP$,
as well as a few of its extensions and special cases, we are ready
to define what constitutes a solution to this problem.

\begin{definition}
An {\em assignment} for an instance $(N, A, P)$ of $\GASP$ 
is a mapping $\pi: N \rightarrow A$;
$\pi(i)=\aes$ means that agent $i$ does not participate in any activity.
Each assignment naturally partitions the agents into at most $|A|$
groups: we set $\pi^0=\{i\mid \pi(i)=\aes\}$
and $\pi^j=\{i\mid \pi(i)=a_j\}$ for $j=1, \dots, p$. 
Given an assignment $\pi$, the {\em coalition structure} 
$\mathit{CS}_\pi$ induced by $\pi$ is the coalition structure over $N$ defined as follows: 
$$
\mathit{CS}_\pi = \left\{\pi^j\mid j=1, \dots, p, \pi^j\neq\emptyset\right\}\cup
                  \left\{\{i\}\mid i\in \pi^0\right\}.
$$
%
\end{definition}

Clearly, not all assignments are equally desirable. As a minimum requirement,
no agent should be assigned to a coalition that he deems unacceptable.
More generally, we prefer an assignment to be stable, 
i.e., no agent (or group of agents) should have an incentive to change its activity.
Thus, we will now define several {\em solution concepts}, i.e., classes
of desirable assignments. We will state our definitions for $\aGASP$ only, 
though all of them can be extended to the more general case of $\GASP$ in a natural way.
Given the connection to hedonic games pointed out in 
Section~\ref{sec:hedonic}, we will proceed by adapting 
the standard hedonic game solution concepts to our setting; however, this has
to be done carefully to preserve intuition that is specific to our model.
 
The first solution concept that we will consider is {\em individual rationality}.
\begin{definition}
Given an instance $(N, A, P)$ of $\aGASP$, 
an assignment $\pi:N\to A$ is said to be  {\em individually rational} if for
every $j>0$ and every agent $i\in \pi^j$ it holds that $(a_j, |\pi^j|)\in S_i$.
\end{definition}
Clearly, if an assignment is not individually rational, there exists
an agent that can benefit from abandoning his coalition in favor of the void activity.
Further, an individually rational assignment always exists: for instance, 
we can set $\pi(i)=\aes$ for all $i\in N$. However, a benevolent central authority
would usually want to maximize the number of agents that are assigned to non-void
activities. Formally, let $\#(\pi)=|\{i\mid \pi(i)\neq \aes\}|$ denote the number
of agents assigned to a non-void activity. We say that $\pi$ is
{\em maximum individually rational} if $\pi$ is individually rational and 
$\#(\pi)\ge \#(\pi')$ for every individually rational assignment $\pi'$.
Further, we say that $\pi$ is {\em perfect}\footnote{The terminological similarity with
the notion of perfect partition in a hedonic game~\cite{AzizBH11} is not a coincidence; there a perfect
partition assigns each agent to her preferred coalition; here a perfect assignment assigns each
agent to one of her equally preferred alternatives.} 
if $\#(\pi)=n$. 
We denote the size of a maximum individually rational
assignment for an instance $(N, A, P)$ by $\#(N, A, P)$.
In Section~\ref{complexity}, we study the complexity of computing a 
perfect or maximum individually rational assignment  for $\aGASP$, 
both for the general model and for the special cases defined in Section~\ref{sec:special}. 

Besides individual rationality, there are a number of solution concepts
for hedonic games that aim to capture stability against individual or group
deviations, such as Nash stability, individual stability, contractual individual
stability, and (weak and strong) core stability (see, e.g.,~\cite{coopbook}).
In what follows, due to lack of space, we only provide the formal definition
(and some results) for Nash stability. 
We briefly discuss how to adapt 
other notions of stability to our setting, but we leave 
the detailed study of their algorithmic properties as a topic for future work.

\begin{definition}\label{def:nash}
Given an instance $(N, A, P)$ of $\aGASP$, 
an assignment $\pi:N\to A$ is said to be {\em Nash stable} 
if it is individually rational 
and for every agent $i\in N$ such that $\pi(i) = \aes$ and every $a_j \in A^*$ 
it holds that $(a_j, |\pi^j| + 1) \not\in S_i$. 
\end{definition}
If $\pi$ is not Nash stable, then there is an agent assigned to the void activity
who wants to join a group that is engaged in a non-void activity, i.e., he would
have approved of the size of this group and its activity choice
if he was one of them. Note that a perfect assignment is Nash stable.
The reader can verify that our definition is a direct adaptation
of the notion of Nash stability in hedonic games: if an assignment is individually
rational, the only agents who can profitably deviate are the ones assigned
to the void activity.
The requirement of Nash stability is much stronger 
than that of individual rationality, and there are cases where 
a Nash stable assignment does not exist (the proof is omitted due to space limits).

\begin{proposition}\label{prop:ns-nonexistence}
For each $n\ge 2$, there exists an instance $(N, A, P)$ 
of $\aGASP$ with $|N|=n$ that does not admit a Nash stable assignment.
This holds even if $|A^*|=1$ and all agents have interval preferences.
\end{proposition}


In Definition~\ref{def:nash} an agent is allowed to join a coalition even
if the members of this coalition are opposed to this. In contrast, the notion of 
{\em individual stability} only allows a player to join a group if none
of the existing group members objects.
%
%
We remark that if all agents have increasing preferences, individual stability
is equivalent to Nash stability: no group of players would object
to having new members join. 

A related hedonic games solution concept is
{\em contractual individual stability}: under this concept, an agent
is only allowed to move from one coalition to another if neither 
the members of his new coalition nor the members of his old coalition object
to the move. However, for $\aGASP$ contractual
individual stability is equivalent to individual stability. Indeed, 
in our model no agent assigned to a non-void activity has an incentive 
to deviate, so we only need to consider deviations from singleton coalitions. 

The solution concepts discussed so far deal with individual deviations;
resistance to group deviations is captured by the notion of the {\em core}. 
One typically distinguishes between {\em strong} group deviations, which
are beneficial for each member of the deviating group, and {\em weak}
group deviations, where the deviation should be beneficial for at least
one member of the deviating group and non-harmful for others; these notions
of deviation correspond to, respectively, {\em weak} and {\em strong} core.
%
%
We note that in the context of $\aGASP$ strong
group deviations amount to players in $\pi^0$ forming a coalition in order
to engage in a non-void activity. This observation immediately implies
that every instance of $\aGASP$ has a non-empty weak core, and an outcome
in the weak core can be constructed by a natural greedy algorithm;
we omit the details due to space constraints. 



\section{Computing Good Outcomes}\label{complexity}
In this section, we consider the computational complexity of finding
a ``good'' assignment for $\aGASP$.
We mostly focus on finding perfect or maximum 
individually rational assignments; towards the end of the section, we also consider
Nash stability. Besides the general case of our problem, 
we consider special cases obtained by combining
constraints on the number and type of activities (e.g., unlimited number
of simple activities, a constant number of copyable activities, etc.)
and constraints on voters' preferences (INC, DEC, INV, etc.). 
Note that if we can find a maximum individually rational assignment, 
we can easily check if a perfect assignment exists, by looking at the size
of our maximum individually rational assignment. Thus, 
we will state our hardness results for the ``easier'' perfect assignment problem
and phrase our polynomial-time algorithms in terms of the ``harder''
problem of finding a  maximum individually rational assignment.


\subsection{Individual Rationality: Hardness Results}\label{sec:hard}
We start by presenting four $\np$-completeness results, which show that 
finding a perfect assignment is hard even under fairly strong constraints
on preferences and activities. We remark that this problem is obviously in $\np$, 
so in what follows we will only provide the hardness proofs.

Our first hardness result applies when 
all activities are simple and the agents' preferences are increasing.

\begin{theorem}\label{npc:simple-inc}
It is $\np$-complete to decide whether $\aGASP$ admits a perfect assignment, even when
all activities in $A^*$ are simple and all agents have increasing preferences.
\end{theorem}
\begin{proof}[sketch]
We provide a reduction from {\sc Exact Cover by 3-Sets (X3C)}.
Recall that an instance of {\sc X3C} is a pair $\langle X, \calY\rangle$,
where $X = \{1, \ldots, 3q\}$ and
$\calY = \{Y_1,\ldots, Y_p\}$ is a collection of 3-element subsets of $X$;
it is a ``yes''-instance if $X$ can be covered by exactly $q$ sets from $\calY$,
and a ``no''-instance otherwise.
Given an instance $\langle X, \calY\rangle$ of {\sc X3C}, we construct
an instance of $\aGASP$ as follows.
We set $N = \{1, \ldots, 3q\}$ and
$A^* = \{a_1, \ldots, a_p\}$.
For each agent $i$, we define his vote $V_i$ so that the induced approval vote $S_i$
is given by $S_i = \{ (a_j, k) \mid i \in Y_j, k \geq 3 \}$, and let $P=(V_1, \dots, V_n)$.
Clearly, $(N, A, P)$ is an instance of $\aGASP$ with increasing preferences.
It is not hard to check that $\langle X, \calY\rangle$
is a ``yes''-instance of {\sc X3C} if and only if $(N, A, P)$
admits a perfect assignment.~\qed
\end{proof}
Our second hardness result applies to simple activities and decreasing preferences, 
and holds even if each agent is willing to share each activity with at most one other agent.

\begin{theorem}\label{npc:simple-dec}
It is $\np$-complete to decide whether $\aGASP$ admits a perfect assignment, even when
all activities in $A^*$ are simple, all agents have decreasing preferences, and, 
moreover, for every agent $i\in N$ and every alternative $a\in A^*$ we have
$S_i^{\downarrow a} \subseteq \{1, 2\}$.
\end{theorem} 
\begin{proof}[sketch]
Consider the following restricted variant of the problem of scheduling on unrelated machines.
There are $n$ jobs and $p$ machines.   
An instance of the problem is given by a collection
of numbers $\{p_{ij}\mid i=1, \dots, n, j=1, \dots, p\}$, where
$p_{ij}$ is the running time of job $i$ on machine $j$, and $p_{ij}\in\{1, 2, +\infty\}$
for every $i=1, \dots, n$ and every $j=1, \dots, p$. It is a ``yes''-instance if
there is a mapping $\rho:\{1, \dots, n\}\to\{1, \dots, p\}$ assigning jobs  
to machines so that the makespan is at most $2$,
i.e., for each $j=1, \dots, p$ it holds that $\sum_{i: \rho(i)=j}p_{ij}\le 2$.
This problem is known to be $\np$-hard
(see the proof of Theorem~5 in~\cite{LenstraST90}).   

Given an instance $\{p_{ij}\mid i=1, \dots, n, j=1, \dots, p\}$ of this problem,
we construct an instance of $\aGASP$ as follows. We set $N=\{1, \dots, n\}$,
$A^* = \{a_1, \ldots, a_p\}$. Further, for each agent $i\in N$ we construct
a vote $V_i$ so that the induced approval vote $S_i$ satisfies
$S_i^{\downarrow a_j}=\{1\}$ if $p_{ij}=2$,
$S_i^{\downarrow a_j}=\{1, 2\}$ if $p_{ij}=1$, and
$S_i^{\downarrow a_j}=\emptyset$ if $p_{ij}=+\infty$.
Clearly, these preferences satisfy the constraints in the statement of the theorem, 
and it can be shown that a perfect assignment for $(N, A, P)$ corresponds
to a schedule with makespan of at most $2$, and vice versa.~\qed
\end{proof}
Our third hardness result also concerns simple activities and decreasing preferences.
However, unlike Theorem~\ref{npc:simple-dec}, it holds even if each agent
approves of at most~$3$ activities. The proof proceeds by a reduction
from {\sc Monotone 3-SAT}.

\begin{theorem}\label{npc:simple-dec2}
It is $\np$-complete to decide whether $\aGASP$ admits a perfect assignment, even when
all activities in $A^*$ are simple, all agents have decreasing preferences, and,
moreover, for every agent $i\in N$ it holds that 
$|\{a\mid S_i^{\downarrow a}\neq \emptyset\}|\le 3$.
\end{theorem}
Our fourth hardness result applies even when
there is only one activity, which is $\infty$-copyable, 
and every agent approves at most two alternatives; 
however, the agents' preferences
constructed in our proof do not satisfy any of the structural constraints 
defined in Section~\ref{sec:special}. The proof proceeds by a reduction from {\sc X3C}.

\begin{theorem}\label{npc:copyable}
It is $\np$-complete to decide whether $\aGASP$ admits a perfect assignment, even when
all activities in $A^*$ are equivalent
(i.e., $A^*$ consists of a single $\infty$-copyable activity $a$) 
and for every $i\in N$ we have $|S_i^{\downarrow a}|\le 2$.
\end{theorem}


\subsection{Individual Rationality: Easiness Results}\label{sec:easy}
The hardness results in Section~\ref{sec:hard} imply that if $A^*$
contains an unbounded number of distinct activities, finding a maximum
individually rational assignment is computationally hard, even under
strong restrictions on agents' preferences (such as INC or DEC). 
Thus, we can only hope to develop an efficient algorithm for this problem
if we assume that the total number of activities is small (i.e., bounded
by a constant) or, more liberally, that the number of 
pairwise non-equivalent activities is small, 
and the agents' preferences satisfy additional constraints. 
We will now consider both of these settings, starting with the case
where $p=|A^*|$ is bounded by a constant.

\begin{theorem}\label{thm:easy-fixedK}
There exist an algorithm that given an instance of $\aGASP$
finds a maximum individually rational assignment and runs in time $(n+1)^p\poly(n)$.
\end{theorem}
\begin{proof}
We will check, for each $r=0, \dots, n$, if there is an individually rational
assignment $\pi$ with $\#(\pi)=r$, and output the maximum value of $r$ 
for which this is the case.
Fix an $r\in\{0, \dots, n\}$.
For every vector $(n_1, \dots, n_p)\in\{0, \dots, n\}^p$ that
satisfies $n_1+\dots+n_p=r$
we will check if there exists an assignment of agents to activities
such that for each $j=1, \dots, p$ exactly $n_j$ agents are assigned to
activity $a_j$ (with the remaining agents being assigned to the void activity), 
and each agent approves of the resulting assignment. Each check will take $\poly(n)$ steps, 
and there are at most $(n+1)^p$ vectors to be checked;
this implies our bound on the running time of our algorithm.

For a fixed vector $(n_1, \dots, n_p)$, we construct an instance of
the network flow problem as follows.
Our network has a source $s$, a sink $t$, a node $i$ for each player
$i=1, \dots, n$, and a node $a_j$
for each $a_j\in A^*$. There is an arc of unit capacity from $s$ to
each agent, and an arc of capacity $n_j$
from node $a_j$ to the sink. Further, there is an arc of unit capacity
from $i$ to $a_j$ if and only if $(a_j, n_j)\in S_i$.
It is not hard to see that an integral flow $F$ of size $r$ in this network
corresponds to an individually rational assignment of size $r$.
It remains to observe that it can be checked in polynomial time
whether a given network admits a flow of a given size.~\qed
\end{proof}
Moreover, when $A^*$ consists of a single simple activity $a$, 
a maximum individually rational assignment can be found by a straightforward greedy algorithm.

\begin{proposition}\label{prob_max_na_ss}
Given an instance $(N, A, P)$ of $\aGASP$ with $A^*=\{a\}$, we can find a maximum
individually rational assignment for $(N, A, P)$ in time $O(s\log s)$, 
where $s=\sum_{i\in N}|S_i|$.
\end{proposition}
\begin{proof}
Clearly, $(N, A, P)$ admits an individually rational assignment $\pi$ with $\#(\pi)=k$
if and only if $|\left\{i\mid (a,k)\in S_i\right\}|\ge k$.
Let $\calR=\{(i, k)\mid (a, k)\in S_i\}$; note that $|\calR|=s$.
We can sort the elements of $\calR$ in descending order with respect 
to their second coordinate in time $O(s\log s)$. 
Now we can scan $\calR$ left to right in order to find the largest value of $k$
such that $\calR$ contains at least $k$ pairs that have $k$ as their second coordinate;
this requires a single pass through the sorted list.~\qed 
\end{proof}
Now, suppose that $A^*$ contains many activities, but most of them are equivalent to each
other; for instance, $A^*$ may consist of a single $k$-copyable activity, 
for a large value of $k$. Then the algorithm described in the proof 
of Theorem~\ref{thm:easy-fixedK} is no longer efficient, but
this setting still appears to be more tractable than the one with many distinct activities.
Of course, by Theorem~\ref{npc:copyable}, in the absence of any restrictions on the agents'
preferences, finding a maximum individually rational assignment is hard even for a single
$\infty$-copyable activity. However, we will now show that this problem becomes easy
if we additionally assume that the agents' preferences are increasing or decreasing.

Observe first that for increasing preferences having multiple copies
of the same activity is not useful: if there is an individually rational assignment
where agents are assigned to multiple copies of an activity, 
we can reassign these agents to a single copy of this activity without violating
individual rationality. Thus, we obtain the following easy corollary 
to Theorem~\ref{thm:easy-fixedK}.

\begin{corollary}\label{thm:easy-fixedK-inc}
Let $(N, A, P)$ be an instance of $\aGASP$ with increasing preferences where
$A^*$ contains at most $K$ activities that are not pairwise equivalent. 
Then we can find a maximum individually rational assignment 
for $(N, A, P)$ in time $n^K\poly(n)$.
\end{corollary}
If all preferences are decreasing, 
we can simply eliminate all $\infty$-copyable activities. Indeed, 
consider an instance $(N, A, P)$  of $\aGASP$ where some activity 
$a\in A^*$ is $\infty$-copyable. Then we can assign each agent $i\in N$ 
such that $(a, 1)\in S_i$ to his own copy of $a$; 
clearly, this will only simplify the problem of assigning the 
remaining agents to the activities. 

It remains to consider the case where the agents' preferences are decreasing, 
there is a limited number of copies of each activity, and the number of distinct
activities is small. While we do not have a complete solution for this case, 
we can show that in the case of a single $k$-copyable activity a natural greedy
algorithm succeeds in finding a maximum individually rational assignment.

\begin{theorem}\label{thm:k-copyable-greedy}
Given a decreasing instance $(N, A, P)$ of $\aGASP$ where $A^*$ consists
of a single $k$-copyable activity (i.e., $A^*=\{a_1, \dots, a_k\}$, 
and all activities in $A^*$ are pairwise equivalent), 
we can find a maximum individually rational assignment in time $O(n\log n)$.
\end{theorem}
\begin{proof}
Since all activities in $A^*$ are pairwise equivalent, we can associate
each agent $i\in N$ with a single number $u_i\in\{0, \dots, n\}$, which
is the maximum size of a coalition assigned to a non-void activity 
that he is willing to be a part of.
We will show that our problem can be solved by a simple greedy
algorithm. Specifically, we sort the agents in non-increasing order of $u_i$s. 
From now on, we will assume without loss of generality that $u_1\ge\dots\ge u_n$.
To form the first group, we find the largest value of $i$ such that
$u_i\ge i$, and assign agents $1, \dots, i$
to the first copy of the activity. In other words, we continue adding
agents to the group as long as the agents are happy to join.
We repeat this procedure with the remaining agents until either $k$
groups have been formed
or all agents have been assigned to one of the groups, whichever
happens earlier.

Clearly, the sorting step is the bottleneck of this procedure, so the
running time of our algorithm is $O(n\log n)$.
It remains to argue that it produces a maximum individually rational
assignment. To show this, we start with an arbitrary
maximum individually rational assignment $\pi$ and transform it into
the one produced by our algorithm
without lowering the number of agents that have been assigned to a non-void
activity. We will assume without loss of generality
that $\pi$ assigns all $k$ copies of the activity (even though this is
is not necessarily the case for the greedy algorithm).

First, suppose that $\pi(i) = a_\emptyset$, $\pi(j) = a_\ell$ 
for some $i<j$ and some $\ell\in\{1, \dots, k\}$.
Then we can modify $\pi$ by setting $\pi(i)=a_\ell$,
$\pi(j)=a_\emptyset$. Since $i<j$ implies $u_i\ge u_j$, the modified
assignment is individually rational. By applying this operation
repeatedly, we can assume that the set of agents
assigned to a non-void activity forms a contiguous prefix of 
$1, \dots, n$.

Next, we will ensure that for each $\ell=1, \dots, k$ the group of
agents that are assigned to $a_\ell$ forms
a contiguous subsequence of $1, \dots, n$. 
To this end, let us sort the coalitions in $\pi$ according to their size, 
from the largest to the smallest, breaking ties arbitrarily.
That is, we reassign the $k$ copies of our activity to coalitions in $\pi$ so that
$\ell<r$ implies $|\pi^\ell|\ge |\pi^r|$. Now, suppose that there exist
a pair of players $i, j$ such that $i<j$, $\pi(i)=a_\ell$, $\pi(j)=a_r$,
and $\ell>r$ (and hence $|\pi^\ell|\le |\pi^r|$). We have
$u_j\ge |\pi^r|\ge |\pi^\ell|$, $u_i\ge u_j \ge |\pi^r|$, 
so if we swap $i$ and $j$ (i.e., modify $\pi$ by setting
$\pi(j)=a_\ell$, $\pi(i) = a_r$), the resulting assignment
remains individually rational. Observe that every such swap
increases the quantity $\Sigma = \sum_{t=1}^k\sum_{s\in\pi^t}(s\cdot t)$
by at least $1$: prior to the swap, the contribution of $i$ and $j$
to $\Sigma$ is $i\ell+jr$, ans after the swap it is $ir+j\ell > i\ell+jr$.
Since for any assignment we have $\Sigma\le kn(n+1)/2$, eventually
we arrive to an assignment where no such pair $(i, j)$ exists.
At this point, each $\pi^\ell$, $\ell=1, \dots, k$, forms 
a contiguous subsequence of $1, \dots, n$, 
and, moreover, $\ell<r$ implies $i\le j$ for all $i\in\pi^\ell$, $j\in\pi^r$.

Now, consider the smallest value of $\ell$ such that $\pi^\ell$
differs from the $\ell$-th coalition constructed by the greedy
algorithm (let us denote it by $\gamma^\ell$),
and let $i$ be the first agent in $\pi^{\ell+1}$. The description of
the greedy algorithm implies that $\pi^\ell$
is a strict subset of $\gamma^\ell$ and agent $i$ belongs to
$\gamma^\ell$. Thus, if we modify $\pi$ by moving
agent $i$ to $\pi^\ell$, the resulting allocation remains individually
rational (since $i$ is happy in $\gamma^\ell$).
By repeating this step, we will gradually transform $\pi$ into the
output of the greedy algorithm (possibly discarding
some copies of the activity). This completes the proof.~\qed\vspace{-2mm}
\end{proof}
The algorithm described in the proof of Theorem~\ref{thm:k-copyable-greedy}
can be extended to the case where we have one $k$-copyable activity $a$
and one simple activity $b$, and the agents have decreasing preferences 
over both activities. For each $s=1, \dots, n$ we will look for
the best solution in which $s$ players are assigned to $b$;
we will then pick the best of these $n$ solutions.
For a fixed $s$ let $N_s=\{i\in N\mid (b, s)\in S_i\}$. 
If $|N_s|<s$, no solution for this value of $s$ exists.
Otherwise, we have to decide which size-$s$
subset of $N_s$ to assign to $b$.
It is not hard to see that we should simply pick the agents
in $N_s$ that have the lowest level of tolerance for $a$,
i.e., we order the agents in $N_s$ by the values of $u^a_i$
from the smallest to the largest, and pick the first $s$ agents.
We then assign the remaining agents to copies of $a$
using the algorithm from the proof of Theorem~\ref{thm:k-copyable-greedy}.
Indeed, any assignment can be transformed into
one of this form by swapping agents so that the individual rationality
constraints are not broken. It would be interesting to see if this idea
can be extended to the case where instead of a single simple activity
$b$ we have a constant number of simple activities 
or a single $k'$-copyable activity.

We conclude this section by giving an $O(\sqrt{n})$-approximation
algorithm for finding a maximum individually rational assignment
in $\aGASP$ with a single $\infty$-copy\-able activity.

\begin{theorem}
  There exists a polynomial-time algorithm that given an instance $(N, A, P)$ of $\aGASP$
  where $A^*$ consists of a single $\infty$-copyable activity $a$, 
  outputs an individually rational assignment $\pi$ with 
  $\#(\pi)=\Theta(\frac{1}{\sqrt{n}})\#(N, A, P)$.
\end{theorem}
\begin{proof}
  We say that an agent $i$ is {\em active} in $\pi$ if $\pi(i)\neq\aes$;
  a coalition of agents is said to be {\em active} if it is assigned 
  to a single copy of $a$.
  We construct an individually rational assignment $\pi$ iteratively,  
  starting from the assignment where no agent is active. 
  Let $N^*=\{i\mid \pi(i)=\aes\}$ be the current set of inactive agents
  (initially, we set $N^*=N$). 
  At each step, we find the largest subset of $N^*$ that can be assigned
  to a single fresh copy of $a$ without breaking the individual rationality constraints, 
  and append this assignment to $\pi$. 
  We repeat this step until the inactive agents cannot form another coalition.
  
  Now we compare the number of active agents in $\pi$ with the number of active agents 
  in a maximum individually rational assignment $\pi^*$. 
  To this end, let us denote the active coalitions of $\pi$
  by $B_1,\dots,B_s$, where $|B_1|\ge\ldots\ge|B_s|$. 
  If $|B_1|\ge \sqrt{n}$, we are done, so assume that this is not the case.  
  Note that since $B_1$ was chosen greedily, 
  this implies that $|C|\le \sqrt{n}$ for every active coalition $C$ in $\pi^*$.

  Let $\mathcal{C}$ be the set of active coalitions in $\pi^*$.
  We partition $\mathcal{C}$ into $s$ groups 
  by setting $\mathcal{C}^1 = \{C\in \mathcal{C}\mid C\cap B_1\neq\emptyset\}$
  and $\mathcal{C}^i = \{C\in\mathcal{C}\mid C\cap B_i\neq\emptyset, 
                         C\not\in \mathcal{C}^j\text{ for }j<i\}$ for $i=2, \dots, s$.
  Note that every active coalition $C\in\pi^*$ intersects some coalition in $\pi$:
  otherwise we could add $C$ to $\pi$. Therefore, each active coalition in $\pi^*$
  belongs to one of the sets $\mathcal{C}^1, \dots, \mathcal{C}^s$. Also, 
  by construction, the sets $\mathcal{C}^1, \dots, \mathcal{C}^s$ are pairwise disjoint. 
  Further, since the coalitions in $\mathcal{C}^i$ are pairwise disjoint
  and each of them intersects $B_i$, we have $|\mathcal{C}^i|\le |B_i|$ for each $i=1, \dots, s$.
  Thus, we obtain 
  \begin{eqnarray*}
  \#(\pi^*) & = & \sum_{i=1, \dots, s} \sum_{C\in\mathcal{C}^i}|C| 
             \le  \sum_{i=1, \dots, s} \sum_{C\in\mathcal{C}^i}\sqrt{n}\\ 
            &\le& \sum_{i=1, \dots, s} |\mathcal{C}^i|\sqrt{n} 
             \le  \sum_{i=1, \dots, s} |B_i|\sqrt{n} 
             \le \#(\pi)\sqrt{n}. \qquad\qquad\qquad\qed
  \end{eqnarray*}
\end{proof}


\subsection{Nash Stability}\label{sec:nash}

We have shown that $\aGASP$ does not not always admit a Nash stable assignment
(Proposition~\ref{prop:ns-nonexistence}). In fact, 
it is difficult to determine whether a Nash stable assignment exists.
The proofs of the next two results are omitted due to space constraints.

\begin{theorem}\label{nash-npc}
It is $\np$-complete to decide whether $\aGASP$ admits a Nash stable assignment. 
\end{theorem}
However, if agents' preferences satisfy INC, DEC, or MIX, 
a Nash stable assignment always exists and can be computed efficiently.

\begin{theorem}\label{thm:nash-in-dec-mix}
If $(N, A, P)$ is an instance of $\aGASP$ that is increasing, decreasing, or mixed
increasing-decreasing, a Nash stable assignment always exists and can be found in polynomial time.
\end{theorem}
Moreover, the problem of finding a Nash stable assignment
that maximizes the number of agents assigned to a non-void activity
admits an efficient algorithm 
if $A^*$ consists of a single simple activity. 

\begin{theorem}\label{prop:max-nash-1act}
There exist a polynomial-time algorithm that given an instance $(N, A, P)$ of $\aGASP$
with $A^*=\{a\}$ finds a Nash stable assignment maximizing 
the number of agents assigned to a non-void activity, 
or decides that no Nash stable assignment exists.
\end{theorem}
\begin{proof}
For each $k=n, \dots, 0$, we will check if there exists a Nash
stable assignment $\pi$ with $\#(\pi)=k$, and output the largest value of 
$k$ for which this is the case. 

For each $i\in N$, let $S'_i=S^{\downarrow a}_i$.
For $k=n$ a Nash stable assignment $\pi$ with $\#(\pi)=n$ exists
if and only if $n\in S'_i$ for each $i\in N$. Assigning every agent to $\aes$
is Nash stable if and only if $1\notin S'_i$ for each $i\in N$. 
Now we assume $1\le k\le n-1$ and set 
$U_1 =\{ i\in N \mid k\in S'_i,    k+1\notin S'_i\}$, 
$U_2 =\{ i\in N \mid k\notin S'_i, k+1\in S'_i\}$, and 
$U_3 =\{ i\in N \mid k\in S'_i,    k+1\in S'_i\}$.
If $|U_1|+|U_3|<k$, there does not exist an individually rational
assignment $\pi$ with $\#(\pi)=k$. If $U_2\neq \emptyset$, no Nash
stable assignment $\pi$ with $\#(\pi)=k$ can exist, since each agent from
$U_2$ would want to switch. 
If $|U_3|>k$, no Nash stable assignment $\pi$ with $\#(\pi)=k$
can exist, since at least one agent in $U_3$ would not
be assigned to $a$ and thus would be unhappy. Finally, if
$|U_1|+|U_3|\ge k$, $|U_3|\le k$, $U_2=\emptyset$, we can construct
a Nash stable assignment $\pi$ by assigning all agents from $U_3$ and $k-|U_3|$
agents from $U_1$ to $a$. Since we have $\pi(i)= a_\emptyset$
for all $i$ with $k\not\in S'_i$ and
$\pi(i)\neq a_\emptyset$
for all $i$ with $k+1\in S'_i$, no agent is unhappy.~\qed
\end{proof}

%
%



\vspace{-4mm}

\section{Conclusions and Future Work}\label{discussion}

We have defined a new model for the selection of a number of group activities, discussed its connections with 
hedonic games, defined several stability notions, and, for two of them, we have obtained several complexity
results. A number of our results are positive: finding desirable assignments proves to be tractable
for several restrictions of the problem that are meaningful in practice. 
Interesting directions for future work include exploring the complexity of computing other solution concepts
for $\aGASP$ and extending our results to the more general setting of $\GASP$.

\smallskip

\noindent{\bf Acknowledgments\ }
This research was supported by National Research Foundation (Singapore) under Research
Fellowship NRF2009-08, by the project ComSoc (ANR-09-BLAN-0305-01), 
and by the Austrian Science Fund (P23724-G11 and P23829-N13). This project was 
initiated during the Dagstuhl seminar 12101 ``Computation and Incentives in Social Choice'',
and the authors are very grateful to Dagstuhl for providing a great research environment 
and inspiration for this work.  We thank Craig Boutilier and Michel Le Breton for helpful comments. 
Part of this work was done when the second author was visiting Universit\'e Paris-Dauphine.


\begin{thebibliography}{1}
\bibitem{AzizBH11}
H.~Aziz, F.~Brandt, and P.~Harrenstein.
\newblock Pareto optimality in coalition formation.
\newblock In {\em SAGT}, pages 93--104, 2011.

\bibitem{ballester}
C.~Ballester.
\newblock {N}{P}-competeness in hedonic games.
\newblock {\em Games and Economic Behavior}, 49:1--30, 2004.

\bibitem{banerjee:1998a}
S.~Banerjee, H.~Konishi, and T.~S{\"o}nmez.
\newblock Core in a simple coalition formation game.
\newblock {\em Social Choice and Welfare}, 18:135--153, 2001.

\bibitem{bogomolnaia:2002}
A.~Bogomolnaia and M.~O. Jackson.
\newblock The stability of hedonic coalition structures.
\newblock {\em Games and Economic Behavior}, 38:201--230, 2002.

\bibitem{coopbook}
G.~Chalkiadakis, E.~Elkind, and M.~Wooldridge.
\newblock {\em Computational Aspects of Cooperative Game Theory}.
\newblock Morgan and Claypool, 2011.

\bibitem{LenstraST90}
J.~Lenstra, D.~Shmoys, and {\'E}.~Tardos.
\newblock Approximation algorithms for scheduling unrelated parallel machines.
\newblock {\em Math. Program.}, 46:259--271, 1990.

\bibitem{LuBoutilier12}
T.~Lu and C.~Boutilier.
\newblock Matching models for preference-sensitive group purchasing.
\newblock In {\em ACM Conference on Electronic Commerce}, pages 723--740, 2012.

\end{thebibliography}
\end{document}